\documentclass[a4paper, 9pt, twosides]{amsart}
\usepackage{amsaddr}
\usepackage{amsmath, amsfonts, amssymb, mathtools, url, graphicx, graphicx}
\usepackage{newtxtext}
\usepackage{newtxmath}

\usepackage{enumitem}
\usepackage[usenames, dvipsnames]{color}
\usepackage[colorlinks=true,
        raiselinks=true,
        linkcolor=MidnightBlue,
        urlcolor=PineGreen,
        citecolor=Brown,
        pdfauthor={Atreyee},
        pdftitle={Switched Linear Systems},
        pdfkeywords={continuous-time, t multiply and divide},
        pdfsubject={Technical Report},
        plainpages=false]{hyperref}

\allowdisplaybreaks
\usepackage{tikz}
\usetikzlibrary{automata,positioning}
\newtheorem{theorem}{Theorem}

\newtheorem{lemma}{Lemma}
\newtheorem{proposition}{Proposition}
\newtheorem{defn}{Definition}
\newtheorem{assump}{Assumption}

\newtheorem{prob}{Problem}
\newtheorem{remark}{Remark}

\newcommand{\norm}[1]{\left\lVert{#1}\right\rVert}
\newcommand{\abs}[1]{\left\lvert{#1}\right\rvert}

\newcommand{\nn}{\nonumber}

\newcommand{\pmat}[1]{\begin{pmatrix}#1\end{pmatrix}}

\newcommand{\R}{\mathbb{R}}
\newcommand{\N}{\mathbb{N}}
\renewcommand{\P}{\mathcal{P}}

\renewcommand{\S}{\mathcal{S}}

\newcommand{\KL}{\mathcal{KL}}
\newcommand{\Kinfty}{\mathcal{K}_{\infty}}

\newcommand{\G}{\mathcal{G}}
\newcommand{\V}{\mathcal{V}}

\newcommand{\C}{\mathcal{C}}

\newcommand{\SR}{\mathcal{S}_{\mathcal{R}}}

\newcommand{\Nsw}{\mathrm{N}}
\newcommand{\Tsw}{\mathrm{T}}

\allowdisplaybreaks

\title[Stabilizing switched systems]{Stabilizing switched nonlinear systems\\under restricted but arbitrary switching signals}
\author{Atreyee Kundu}
\thanks{The author is with the Department of Electrical Engineering, Indian Institute of Technology Kharagpur, West Bengal - 721302, India, email: atreyee@ee.iitkgp.ac.in. }


\date{\today}

\begin{document}

	\maketitle

	\begin{abstract}
        This paper deals with input/output-to-state stability (IOSS) of switched nonlinear systems whose switching signals obey pre-specified restrictions on admissible switches between the subsystems and admissible dwell times on the subsystems. We present sufficient conditions on the subsystems, admissible switches between them and admissible dwell times on them, such that a switched system generated under all switching signals obeying the given restrictions is IOSS. Multiple Lyapunov-like functions and graph theory are the key apparatuses for our analysis. A numerical example is presented to demonstrate our results. 
	\end{abstract}
	
\section{Introduction}
\label{s:intro}
    A \emph{switched system} has two ingredients --- a family of systems and a switching signal. The \emph{switching signal} selects an \emph{active subsystem} at every instant of time, i.e., the system from the family that is currently being followed \cite[Section 1.1.2]{Liberzon}. Such systems find wide applications in power systems and power electronics, automotive control, aircraft and air traffic control, network and congestion control, etc. \cite[p.\ 5]{Sun}. 
    
    In this paper we will study input/output-to-state stability (IOSS) of continuous-time switched nonlinear systems whose switching signals obey pre-specified restrictions on admissible switches between the subsystems and admissible dwell times on the subsystems. Such restrictions on switching signals arise in many engineering applications, see e.g., \cite[Remark 1]{abc} for examples. IOSS property of a switched system implies that irrespective of the initial state, if the inputs and the observed outputs are small, then the state of the system will become small eventually. This property is useful in the design of state-norm estimators for switched systems. 
     
    IOSS of switched systems is studied in the literature broadly in three directions --- stability under arbitrary switching, stability under classes of (average) dwell time switching signals, and stability under switching signals that obey pre-specified restrictions.
    IOSS of switched differential inclusions under arbitrary switching is studied in \cite{Mancilla_2005}. The analysis relies on the existence of a common IOSS Lyapunov function. IOSS of a class of hybrid systems that admits a Lyapunov function satisfying an IOSS relation both along the flow and during the jumps, is addressed in \cite{Sanfelice_2010}. 
    IOSS of switched systems under average dwell time switching \cite[Chapter 3]{Liberzon} is studied in \cite{Liberzon_2012}. The authors consider two settings: all subsystems are IOSS and some subsystems are unstable. It is shown that if a switching signal obeys an average dwell time property (resp., an average dwell time property and a constrained activation of unstable subsystems), then IOSS of the resulting switched system is preserved. IOSS of impulsive switched systems is addressed in \cite{Li_2018}. Stability conditions for both stabilizing and destabilizing impulses are presented by employing Lyapunov functions and average dwell time switching conditions. Recently in \cite{Liu_2022a} 
    ISS and an integral version of it are studied using nonlinear supply functions. Sufficient conditions for integral ISS of switched systems with jumps under slow switching is proposed in \cite{Liu_2022b}. 
    In \cite{abc} we present an algorithm to design \emph{a} periodic switching signal that obeys pre-specified restrictions on admissible switches between the subsystems and admissible minimum and maximum dwell times on the subsystems and preserves IOSS of the resulting switched system. The said design involves finding a class of cycles on a weighted directed graph representation of a switched system that satisfies certain conditions involving multiple Lyapunov-like functions and the admissible dwell times. In \cite{def,ghi} we focus on identifying \emph{a class of stabilizing switching signals} that obeys given restrictions. Our characterization of stabilizing switching signals relies on frequency and fraction of activation of various classes of subsystems. 
    
    In this paper we continue with our study of IOSS under switching signals that obey pre-specified restrictions on admissible switches between the subsystems and admissible dwell times on the subsystems. Instead of focussing on stabilizing subclasses of these restricted switching signals or designing a stabilizing element in the class of restricted switching signals, we are interested in stability under every element in the class. We present sufficient conditions on the subsystems dynamics, set of admissible switches between the subsystems and set of admissible dwell times on the subsystems under which every switching signal that obeys the given restrictions preserves IOSS of the resulting switched system.  
    
    Multiple IOSS-Lyapunov-like functions and graph theoretic arguments are employed in our analysis. We associate a weighted directed graph with a family of systems and the admissible switches between them in a natural way. Properties of the subsystems and the switches between them are captured by certain scalars computed from their corresponding IOSS-Lyapunov-like functions and are associated to vertex and edge weights of the graph. The admissible switching signals, i.e., the switching signals that obey the given restrictions on admissible switches and admissible dwell times, are associated to a class of infinite walks on the underlying weighted directed graph of the switched system described above. 
    
    We show that if the subsystems dynamics and the given restrictions are such that the underlying weighted directed graph of the switched system admits a class of finite walks that satisfies a certain property that we call contractivity, then the switched system under consideration is IOSS under all switching signals obeying the given restrictions. In particular, the infinite walks of our interest contain these walks as subwalks and contribute to stabilizing property of their corresponding switching signals. We also provide sufficient conditions under which our stability conditions can be verified in a numerically tractable manner. To the best of our knowledge, this is the first instance in the literature where IOSS of continuous-time switched nonlinear systems under arbitrary elements belonging to a class of switching signals that obeys pre-specified restrictions on admissible switches between the subsystems and admissible dwell times on the subsystems is reported.
    
    The remainder of this paper is organized as follows: in Section \ref{s:prob_stat} we formulate the problem under consideration. We catalogue a set of preliminaries required for our results in Section \ref{s:prelims}. Our results appear in Section \ref{s:mainres}. A numerical example is presented in Section \ref{s:numex}. We conclude in Section \ref{s:concln} with a brief discussion of future research direction. Proofs of our results appear in a consolidated manner in Section \ref{s:proofs}.
    
    We employ standard notation throughout the paper. \(\R\) is the set of real numbers, \(\norm{\cdot}\) is the Euclidean norm, and for any interval \(I\subseteq[0,+\infty[\), \(\norm{\cdot}_{I}\) is the essential supremum norm of a map from \(I\) into some Euclidean space. 
\section{Problem statement}
\label{s:prob_stat}
    We consider a family of continuous-time nonlinear systems
    \begin{align}
    \label{e:family}
    \begin{aligned}
        \dot{x}(t) &= f_p(x(t),v(t)),\\
        y(t) &= h_p(x(t)),
    \end{aligned}
        \:x(0) = x_0,\:p\in\P,\:t\in[0,+\infty[,
    \end{align}
    where \(x(t)\in\R^d\), \(v(t)\in\R^m\) and \(y(t)\in\R^p\) are the vectors of states, inputs and outputs at time \(t\), respectively, \(\P=\{1,2,\ldots,N\}\) is an index set. Let \(\sigma:[0,+\infty[\to\P\) be the \emph{switching signal} --- it is a piecewise constant function that selects at each time \(t\), the index of the active subsystem, i.e., the system from the family \eqref{e:family} that is currently being followed. By convention, \(\sigma\) is assumed to be continuous from right and having limits from the left everywhere. We let \(\S\) denote the set of all switching signals. A switched system generated by the family of systems \eqref{e:family} and a switching signal, \(\sigma\), is given by
    \begin{align}
    \label{e:swsys}
    \begin{aligned}
        \dot{x}(t) &= f_{\sigma(t)}(x(t),v(t)),\\
        y(t) &= h_{\sigma(t)}(x(t)),
    \end{aligned}
    \:x(0)=x_0,\:t\in[0,+\infty[.
    \end{align}
    We assume that for each \(p\in\P\), \(f_p\) is locally Lipschitz, \(f_p(0,0) = 0\) and \(h_p\) is continuous, \(h_p(0) = 0\). The exogenous inputs are Lebesgue measurable and essentially bounded. Thus, a solution to the switched system \eqref{e:swsys} exists in the Carath\'{e}odory sense for some non-trivial time interval containing \(0\) \cite[Chapter 2]{Filippov}. 
    
    We will restrict our attention to elements of the set \(\S\) that obey pre-specified restrictions on admissible switches between the subsystems and admissible dwell times on the subsystems. Let \(0=:\tau_0<\tau_1<\cdots\) be the \emph{switching instants}, i.e., the points of discontinuity of \(\sigma\). We let \(E(\P)\subseteq\P\times\P\) denote the set of admissible switches, i.e., all ordered pairs \((p,q)\) such that a transition from subsystem \(p\) to subsystem \(q\) is admissible, \(p\neq q\). Let \(\delta_p\) and \(\Delta_p\) denote the admissible minimum and maximum dwell times on subsystem \(p\in\P\), respectively. 
    \begin{defn}
    \label{d:adm-sw}
        A switching signal, \(\sigma\in\S\), is \emph{admissible} if it satisfies the following:
            \(\bigl(\sigma(\tau_i),\sigma(\tau_{i+1})\bigr)\in E(\P)\) and \(\tau_{i+1}-\tau_{i}\in[\delta_{\sigma(\tau_i)},\Delta_{\sigma(\tau_i)}]\) for all \(i=0,1,2,\ldots\).
    \end{defn}
    
    Let \(\SR\) be the set of all admissible switching signals, \(\sigma\in\S\). In this paper we will study IOSS of the switched system \eqref{e:swsys} under all the elements of the set \(\SR\).
    \begin{defn}{\cite[Appendix A.6]{Liberzon}}
    \label{d:ioss}
        The switched system \eqref{e:swsys} is input/output-to-state stable (IOSS) under a switching signal \(\sigma\in\S\) if there exist class \(\Kinfty\) functions \(\alpha\), \(\chi_1\), \(\chi_2\) and a class \(\KL\) function \(\beta\) such that for all inputs \(v\) and initial states \(x_0\), we have
        \begin{align}
        \label{e:ioss}
            \alpha(\norm{x(t)})\leq\beta(\norm{x_0},t)+\chi_1(\norm{v}_{[0,t]})+\chi_2(\norm{y}_{[0,t]})
        \end{align}
        for all \(t\geq 0\). If \(\chi_2\equiv0\), then \eqref{e:ioss} reduces to input-to-state stability (ISS) of \eqref{e:swsys} \cite[Appendix A.6]{Liberzon} and if also \(v\equiv 0\), then \eqref{e:ioss} reduces to global asymptotic stability (GAS) of \eqref{e:swsys} \cite[Appendix A.1]{Liberzon}.
    \end{defn}
    
    We will solve the following problem:
    \begin{prob}
    \label{prob:main}
        Find conditions on the family of systems \eqref{e:family}, the set of admissible switches, \(E(\P)\), and the set of admissible minimum and maximum dwell times, \(\delta_p\) and \(\Delta_p\) on the subsystems \(p\in\P\), such that the switched system \eqref{e:swsys} is IOSS under all admissible switching signals, \(\sigma\in\SR\). 
    \end{prob}
    
    We will employ multiple Lyapunov-like functions \cite{Branicky_1998} and graph theory as the main apparatuses for our analysis. Prior to presenting our solution to \ref{prob:main}, we catalog a set of preliminaries.
\section{Preliminaries}
\label{s:prelims}
\subsection{Family of systems \eqref{e:family}}
\label{ss:family}
    Let \(\P_S\) and \(\P_U\) denote the set of indices of IOSS and non-IOSS subsystems, respectively. \(\P=\P_S\sqcup\P_U\).
    \begin{assump}
    \label{assump:key1}
        There exist class \(\Kinfty\) functions \(\underline{\alpha}\), \(\overline{\alpha}\), \(\gamma_1\), \(\gamma_2\), continuously differentiable functions, \(V_p:\R^d\to[0,+\infty[\), \(p\in\P\), and constants \(\lambda_p\in\R\) with \(\lambda_p>0\) for \(p\in\P_S\) and \(\lambda_p<0\) for \(p\in\P_U\), such that for all \(\xi\in\R^d\) and \(\eta\in\R^m\), we have
        \begin{align}
            \label{e:key1}\underline\alpha(\norm{\xi})&\leq V_p(\xi) \leq \overline\alpha(\norm{\xi}),
            \intertext{and}
            \label{e:key2}\left<\frac{\partial V_p}{\partial\xi}(\xi),f_p(\xi,\eta)\right>&\leq-\lambda_p V_p(\xi)+\gamma_1(\norm{\eta})+\gamma_2(\norm{h_p(\xi)}).
        \end{align}
    \end{assump}
    \begin{assump}
    \label{assump:key2}
        For each \((p,q)\in E(\P)\) there exist \(\mu_{pq}\geq 1\) such that the IOSS-Lyapunov-like functions are related as follows:
        \begin{align}
        \label{e:key3}
            V_q(\xi)\leq \mu_{pq}V_p(\xi)\:\text{for all}\:\xi\in\R^d.
        \end{align} 
    \end{assump}
    
    The functions, \(V_p\), \(p\in\P\), are called the IOSS-Lyapunov-like functions. Conditions \eqref{e:key2} is equivalent to the IOSS property for IOSS subsystems \cite{Krichman_2001, Sontag-Wang_1995} and the unboundedness observability property for the non-IOSS subsystems \cite{Angeli-Sontag_1999, Sontag-Wang_1995}. The scalars \(\lambda_p\), \(p\in\P\), provide a quantitative measure of (in)stability of the systems \eqref{e:family}. Condition \eqref{e:key3} implies that the IOSS-Lyapunov-like functions, \(V_p\), \(p\in\P\), are linearly comparable. The scalars, \(\mu_{pq}\), \((p,q)\in E(\P)\), gives an estimate of this comparability.
\subsection{Switching signals, \(\sigma\)}
\label{ss:swsig}
    Fix an interval \(]s,t]\subseteq[0,+\infty[\) of time. We let \(\Nsw(s,t)\) denote the \emph{total number of switches} on \(]s,t]\). Let \(\Tsw_p(s,t)\) and \(\Nsw_{pq}(s,t)\) denote the \emph{total duration of activation of subsystem \(p\in\P\)} and the \emph{total number of switches from subsystem \(p\) to subsystem \(q\), \((p,q)\in E(\P)\)} on \(]s,t]\), respectively. Clearly, \(\displaystyle{t-s=\sum_{p\in\P}\Tsw_p(s,t)}\) and \(\displaystyle{\Nsw(s,t)=\sum_{(p,q)\in E(\P)}\Nsw_{pq}(s,t)}\). 
\subsection{Weighted directed graph and switched system \eqref{e:swsys}}
\label{ss":graph}
    We associate a weighted directed graph, \(\G(\P,E(\P))\), with the family of systems \eqref{e:family} and the set of admissible switches between the subsystems, \(E(\P)\), as follows:
    \begin{itemize}[label=\(\circ\),leftmargin=*]
        \item The index set, \(\P\), is the set of vertices.
        \item The set of edges, \(E(\P)\), consists of a directed edge from vertex \(p\) to vertex \(q\) whenever a switch from subsystem \(p\) to subsystem \(q\) is admissible.
        \item The vertex weights are: \(w(p)=-\abs{\lambda_p}\), \(p\in\P_S\) and \(w(p)=\abs{\lambda_p}\), \(p\in\P_U\), where \(\lambda_p\), \(p\in\P_U\) are as described in Assumption \ref{assump:key1}.
        \item The edge weights are: \(w(p,q) = \ln\mu_{pq}\), \((p,q)\in E(\P)\), where \(\mu_{pq}\), \((p,q)\in E(\P)\) are as described in Assumption \ref{assump:key2}.
    \end{itemize}
    
    In the sequel we will call \(\G(\P,E(\P))\) as the underlying weighted directed graph of the switched system \eqref{e:swsys} and refer to it as \(\G\) whenever there is no risk of confusion. Also, we will switch freely between system theoretic and graph theoretic terminologies in the sense of the above description of \(\G\).
    
    \begin{defn}
    \label{d:walk}
        A \emph{walk} on \(\G\) is a sequence of vertices, \(W = v_0,v_1,\ldots,v_{n-1},v_{n}\), where \(v_i\in\P\), \(i=0,1,\ldots,n\) such that they admit edges \((v_i,v_{i+1})\in E(\P)\), \(i=0,1,\ldots,n-1\). 
    \end{defn}
    \begin{defn}
    \label{d:closed-walk}
        We call \(W= v_0,v_1,\ldots,v_{n-1},v_{n}\) a \emph{closed walk} if \(v_0=v_n\) and \(v_i\neq v_0=v_n\), \(i=1,2,\ldots,n-1\). 
    \end{defn}
    \begin{defn}
    \label{d:simple-walk}
        A walk \(W= v_0,v_1,\ldots,v_{n-1},v_{n}\) is called a \emph{simple walk} if \(v_0\neq v_n\) and \(v_i\), \(i=1,2,\ldots,n-1\) are distinct from each other and \(v_0\), \(v_n\).
    \end{defn}
    \begin{defn}
    \label{d:cycle}
        A closed walk \(W= v_0,v_1,\ldots,v_{n-1},v_{n}\) is called a \emph{cycle} if the vertices \(v_i\), \(i=1,2,\ldots,n-1\) are distinct from each other and \(v_0\).
    \end{defn}
    \begin{defn}
    \label{d:distinct}
        Two walks \(W_1\) and \(W_2\) on \(\G\) are \emph{distinct} if there is at least one vertex \(v\in\P\) that appears in exactly one of the walks. 
    \end{defn}
     For a walk \(W\), we let \(\V(W)\) denote the set of all vertices of \(\G\) that appear in \(W\), and \(\V^c(W)\) denote the set of all vertices of \(\G\) that do not appear in \(W\), i.e., \(\V^c(W) = \P\subset\V(W)\). In the sequel we will occasionally call \(W = v_0,v_1,\ldots,v_{n-1},v_n\), \(v_0\neq v_n\), as a walk (resp., simple walk) from vertex \(v_0\) to vertex \(v_n\) and denote it by \(W_{v_0\to v_n}\). We will call \(W = v_0,v_1,\ldots,v_n\), \(v_0=v_n\) as a closed walk (resp., cycle) on \(v_0\) and denote it by \(W_{v_0\to v_0}\). For a vertex \(v\in\P\), we let \(\C_v\) denote the set of all closed walks on \(v\). 
    \begin{defn}
    \label{d:length}
        The \emph{length of a walk} \(W\) is the number of vertices that appear in \(W\), counting repetitions. In particular, the length of \(W=v_0,v_1,\ldots,v_{n-1},v_n\) is \(n+1\).  By the term \emph{infinite walk}, we refer to a walk of infinite length, i.e., it has infinitely many vertices. For example, \(W = v_0,v_1,\ldots\) is an infinite walk. 
    \end{defn}
    \begin{defn}
    \label{d:subwalk}
        By the term \emph{subwalks} we refer to segments of a walk \(W\). In particular, for a finite walk \(W=v_0,v_1,\ldots,v_{n-1},v_n\), its subwalks are all segments \(v_k,v_{k+1},\ldots,v_{k+m}\) such that \(k\geq 0\) and \(k+m\leq n\), and for an infinite walk \(W=v_0,v_1,\ldots\), its subwalks are the segments \(v_k,v_{k+1},\ldots,v_{k+n}\), \(k\in\{0,1,2,\ldots\}\), \(n\in\{1,2,\ldots\}\). 
    \end{defn}
    \begin{defn}
    \label{d:sw-walk}
        An admissible switching signal \(\sigma\in\SR\) corresponds to an infinite walk \(W=\sigma(\tau_0),\sigma(\tau_1),\sigma(\tau_2),\ldots\) on \(\G\) and satisfies \(\tau_{i+1}-\tau_{i}\in[\delta_{\sigma(\tau_i)},\Delta_{\sigma(\tau_i)}]\), \(i=0,1,2,\ldots\).
    \end{defn}
    
    Infinite walks on \(\G\) whose subwalks have certain properties will be useful for us.
    \begin{defn}
    \label{d:contractive}
        A walk \(W = v_0,v_1,\ldots,v_n\) on \(\G\) is called \emph{contractive} if it satisfies
        \begin{align}
        \label{e:contractive}
            \Xi(W) := \sum_{i=0}^{n-1}w(v_i)D_{v_i}+\sum_{i=0}^{n-1}w(v_i,v_{i+1}) < 0
        \end{align}
        for all \(D_{v_i}\in[\delta_{v_i},\Delta_{v_i}]\), \(i=0,1,\ldots,n-1\).
    \end{defn}
    \begin{defn}
    \label{d:jointly-contractive}
        Two distinct walks \(W_1= v_0^{(1)},v_1^{(1)},\ldots,v_{n_1}^{(1)}\) and \(W_2= v_0^{(2)},v_1^{(2)},\ldots,v_{n_2}^{(2)}\) on \(\G\) with \(v_{n_1}^{(1)} = v_{0}^{(2)}\) are called \emph{jointly contractive} if they satisfy
        \begin{align}
        \label{e:joint-contractive}
            \Xi(W_1)+\Xi(W_2) &:= \sum_{i=0}^{n_1-1}w(v_i^{(1)})D_{v_i^{(1)}}+\sum_{i=0}^{n_1-1}w(v_i^{(1)},v_{i+1}^{(1)})\nonumber\\
            &\hspace*{-0.8cm}+\sum_{i=0}^{n_2-1}w(v_i^{(2)})D_{v_i^{(2)}}+\sum_{i=0}^{n_2-1}w(v_i^{(2)},v_{i+1}^{(2)}) < 0
        \end{align}
        for all \(D_{v_i^{(1)}}\in[\delta_{v_i^{(1)}},\Delta_{v_i^{(1)}}]\), \(i=0,1,\ldots,n_1-1\) and \(D_{v_i^{(2)}}\in[\delta_{v_i^{(2)}},\Delta_{v_i^{(2)}}]\), \(i=0,1,\ldots,n_2-1\).
    \end{defn}
    Notice that the concatenation of two jointly contractive walks is a contractive walk.
    
    \begin{remark}
    \label{rem:all-dwell}
    \rm{
        The association of all admissible dwell times on the vertices that appear in \(W\) with the contractivity (resp., joint contractivity) directly relates to the association of an admissible switching signal, \(\sigma\in\SR\), with an infinite walk on \(\G\). Intuitively, we require stability under all infinite walks \(\sigma(\tau_0),\sigma(\tau_1),\ldots\) with the distance between \(\tau_i\) and \(\tau_{i+1}\) being any choice in the range \([\delta_{\sigma(\tau_i)},\Delta_{\sigma(\tau_i)}]\), \(i=0,1,2,\ldots\). We shall see in Section \ref{s:proofs} that contractivity of certain subwalks of the infinite walks on \(\G\) corresponding to admissible switching signals will contribute to decays in the state trajectory of the switched system \eqref{e:swsys}. 
        }
    \end{remark}
    
    We are now in a position to present our results.
\section{Main results}
\label{s:mainres}
    The following theorem is our solution to Problem \ref{prob:main}:
    \begin{theorem}
    \label{t:mainres}
        Consider the underlying weighted directed graph, \(\G(\P,E(\P))\), of the switched system \eqref{e:swsys}. Suppose that for every \(v\in\P\) with \(\C_v\neq\emptyset\), the following conditions hold:
        \begin{enumerate}[label=C\arabic*), leftmargin = *]
            \item\label{condn1}Every closed walk on \(v\), \(W_{v\to v}^{(1)},W_{v\to v}^{(2)},\ldots,W_{v\to v}^{(r_v)}\) is contractive.
            \item\label{condn2}For every \(u\in\V^{c}_{W_{v\to v}^{(j)}}\), every simple walk \(W_{u\to v}^{(1)},W_{u\to v}^{(2)},\ldots,W_{u\to v}^{(m_{uv})}\) is jointly contractive with \(W_{v\to v}^{(j)}\), \(j=1,2,\ldots,r_v\).
        \end{enumerate}
        Then the switched system \eqref{e:swsys} is input/output-to-state stable (IOSS) under all admissible switching signals, \(\sigma\in\SR\).
    \end{theorem}
    
    Theorem \ref{t:mainres} relies on the properties of the underlying weighted directed graph, \(\G\), of the switched system \eqref{e:swsys} towards guaranteeing IOSS of \eqref{e:swsys} under all switching signals that obey the given restrictions on admissible switches between the subsystems and admissible dwell times on the subsystems. In particular, we require (i) all closed walks \(W_{v\to v}\) on \(\G\) to be contractive and (ii) for each closed walk \(W_{v\to v}\), the simple walks \(W_{u\to v}\) such that \(u\) does not appear in \(W_{v\to v}\), to be jointly contractive with \(W_{v\to v}\). In our solution to Problem \ref{prob:main} given by Theorem \ref{t:mainres}, properties of the subsystems are captured by the scalars, \(\lambda_p\), \(p\in\P\), \(\mu_{pq}\), \((p,q)\in E(\P)\) computed from their corresponding IOSS-Lyapunov-like functions, \(V_p\), \(p\in\P\) and appear as vertex and edge weights of \(\G\), the admissible switches between the subsystems, \(E(\P)\), are captured in admissible walks on \(\G\), and the admissible dwell times on the subsystems, \(\delta_p\), \(\Delta_p\), \(p\in\P\), are captured by weighing factors of the vertex weights in the definition of contractive (resp., jointly contractive) walks. A proof of Theorem \ref{t:mainres} is presented in Section \ref{s:proofs}. 
    If \(\chi_2\equiv 0\), then our result caters to ISS of \eqref{e:swsys} and if also \(v\equiv 0\), then it caters to GAS of \eqref{e:swsys}.
    
    \begin{remark}
    \label{rem:compa1}
    \rm{
        We employed the weighted directed graph, \(\G\), of the switched system \eqref{e:swsys} for the study of IOSS earlier in \cite{abc}. The specific purpose was to design \emph{a} stabilizing periodic switching signal algorithmically. Such a signal was generated by repeating a cycle \(W\) on \(\G\) such that \(W\) satisfies condition \eqref{e:contractive} with an admissible choice of \(D_v\), \(v\in \V(W)\). In the current work we employ \(\G\) for the purpose of ensuring IOSS of \eqref{e:swsys} under all switching signals that obey the given restrictions. Instead of restricting our attention to cycles that obey \eqref{e:contractive} with an admissible choice of \(D_v\), \(v\in \V(W)\), we work with finite walks on \(\G\) that obey \eqref{e:contractive} for \emph{all} admissible choices of \(D_v\), \(v\in \V(W)\). Naturally the use of \(\G\) is generalized in the current work. In particular, our proof of Theorem \ref{t:mainres} to be presented in Section \ref{s:proofs} involves graph theoretic analysis of IOSS of switched system \eqref{e:swsys} beyond periodic switching signals.
        }
    \end{remark}
    
    \begin{remark}
    \label{rem:feature1}
    \rm{
        Notice that by definition of contractive (resp., jointly contractive) walks, we require a sum of edge weights and weighted vertex weights to be negative. The negative values are contributed only by the weights of the IOSS vertices. Thus, our results do not cater to set of admissible switches that allows infinite walks with all unstable vertices.
        }
    \end{remark}
    
    A next natural question is: given the underlying weighted directed graph, \(\G\), of the switched system \eqref{e:swsys}, how do we verify if conditions \ref{condn1} and \ref{condn2} hold? Let us note the difficulties associated to checking these conditions numerically. First, we not have upper bounds on the length of the closed walks. Indeed, multiple vertices can be repeated multiple times in between. Second, by definition of contractive (resp., jointly contractive) walks, certain condition needs to be satisfied for all choices of admissible dwell times on the respective vertices. We provide a set of sufficient conditions under which conditions \ref{condn1} and \ref{condn2} on \(\G\) can be verified numerically.
    \begin{proposition}
    \label{prop:mainres1}
        Consider the underlying weighted directed graph, \(\G(\P,E(\P))\), of the switched system \eqref{e:swsys}. Suppose that the walk \(W = v_0,v_1,\ldots,v_{n-1},v_n\) satisfies
        \begin{align}
        \label{e:aux_condn1}
            \sum_{\substack{i=0\\v_i\in\P_S}}^{n-1}w(v_i)\delta_{v_i} + \sum_{\substack{i=0\\v_i\in\P_U}}^{n-1}w(v_i)\Delta_{v_i} + \sum_{i=0}^{n-1}w(v_i,v_{i+1}) < 0.
        \end{align}
        Then \(W\) is contractive.
    \end{proposition}
    
    The above proposition asserts that to determine contractivity of a walk \(W = v_0,v_1,\ldots,v_{n-1},v_n\), it suffices to check if condition \eqref{e:contractive} holds for \(D_{v_i} = \delta_{v_i}\), if \(v_i\in\P_S\) and \(D_{v_i}=\Delta_{v_i}\), if \(v_i\in\P_U\), \(i=0,1,\ldots,n-1\) and a check for all possible choices of \(D_{v_i}\in[\delta_{v_i},\Delta_{v_i}]\), \(i=0,1,\ldots,n-1\) is not necessary. In particular, given \(\G\), the check for satisfaction of \eqref{e:aux_condn1} can be addressed as a negative weighted walk detection problem. Consider an edge weighted directed graph \(\overline{\G}(\overline{\P},E(\overline{\P}))\) with \(\overline{\P} = \P\), \(E(\overline{\P}) = E(\P)\), and edge weights \(\overline{w}(u,v) = w(u,v) + w(u)\delta_u\), if \(u\in\P_S\) and \(\overline{w}(u,v) = w(u,v) + w(u)\Delta_u\), if \(u\in\P_U\), \((u,v)\in E(\overline{\P})\). Recall that for a negative weighted walk, the sum of edge weights is strictly less than zero. It follows that finding a negative weighted walk on \(\overline{\G}\) is the same as finding a walk on \(\G\) that satisfies condition \eqref{e:aux_condn1}. Intuitively, Proposition \ref{prop:mainres1} can be interpreted as follows: if the positive component of a function resulting from dwelling on unstable subsystems for the maximum admissible duration of time and switching between the subsystems can be contradicted by negative component of the function resulting from dwelling on stable subsystems for the minimum admissible duration of time, then the positive component of the function resulting from dwelling on unstable subsystems for any admissible duration of time and switching between the subsystems can be contradicted the negative component of the function resulting from dwelling on stable subsystems for any admissible duration of time. 
    
    Our next result will solve the problem of unbounded length of the closed walks under consideration.
    \begin{proposition}
    \label{prop:mainres2}
         Consider the underlying weighted directed graph, \(\G(\P,E(\P))\), of the switched system \eqref{e:swsys}. Suppose that for all \(v\in\P\) such that \(\C_v\neq\emptyset\), all cycles \(W_{v\to v}\) are contractive. Then all closed walks \(W_{v\to v}\) are contractive.
    \end{proposition}
    Armed with Propositions \ref{prop:mainres1} and \ref{prop:mainres2}, we can verify conditions \ref{condn1} and \ref{condn2} by checking if \eqref{e:aux_condn1} holds for simple walks and cycles on a given \(\G\). We provide proofs of Propositions \ref{prop:mainres1} and \ref{prop:mainres2} in Section \ref{s:proofs}.
    
    We now present a numerical example to demonstrate the proposed results.
\section{Numerical example}
\label{s:numex}
    We consider a family of systems \eqref{e:family} with \(\P = \{1,2,3\}\), where
    \begin{align*}
        f_1(x,v) &= \pmat{-2x_1+\sin(x_1-x_2)\\-2x_2+\sin(x_2-x_1)+0.5v},\:\:h_1(x) = x_1-x_2,\\
        f_2(x,v) &= \pmat{0.5x_2+0.25\abs{x_1}\\\text{sat}(x_1)+\frac{1}{2}v},\:\:h_2(x) = \abs{x_1},\\
        f_3(x,v) &= \pmat{0.2x_1+0.1x_2\\0.3x_1+v},\:\:h_3(x) = x_1,
    \end{align*}
    where \(\text{sat}(x_1) = \min\{1,\max\{-1,x_1\}\}\). Clearly, \(\P_S = \{1\}\) and \(\P_U=\{2,3\}\). Let the admissible switches be \(E(\P) = \{(1,2),(1,3),(2,1),(3,1)\}\) and the admissible minimum and maximum dwell times be \(\delta_1=\delta_2=\delta_3=3.5\) and \(\Delta_1=\Delta_2=\Delta_3 = 4\).
    
    Let \(V_1(x) = V_2(x) = \frac{1}{2}(x_1^2+x_2^2)\), \(V_3(x) = x_1^2+x_2^2\). We compute \(w(1) = -\abs{\lambda_1}=-3.5\), \(w(2) = \abs{\lambda_2} = w(3) = \abs{\lambda_3} = 0.73\), \(w(1,2) = \ln\mu_{12} = w(2,1) = \ln\mu_{21} = w(3,1) = \ln\mu_{31} = 0\), \(w(1,3) = \ln\mu_{13} = 0.6931\). The underlying weighted directed graph, \(\G\), of the switched system \eqref{e:swsys} is shown in Figure \ref{fig:graph}.
    \begin{figure}
    \begin{center}
            \scalebox{1}{
        \begin{tikzpicture}[every path/.style={>=latex},base node/.style={draw,circle}]
            \node[base node]            (a) at (-1.5,0)  { 1 };
            \node[base node]            (b) at (1.5,0)  { 2 };
            \node[base node]            (c) at (-1.5,-2) { 3 };
            \node (e) at (2.2,0) {$w(2)$};
            \node (f)  at (-2.2,0) {$w(1)$};
            \node (h) at (-2.2,-2) {$w(3)$};
            \node (i)  at  (0,-0.3) {$w(1,2)$};
            \node (j)  at  (0,0.8) {$w(2,1)$};
            \node[rotate=-90] (n) at  (-2.19,-1) {$w(3,1)$};
             \node[rotate=-90] (o) at  (-1.25,-1) {$w(1,3)$};
	   \draw[->] (a) edge (b);
            \draw[->] (b) edge[bend right] (a);

             \draw[->] (c) edge[bend left] (a);

            \draw[->] (a) edge (c);
        \end{tikzpicture}}
    \end{center}
    \caption{Underlying weighted directed graph, \(\G\), for the switched system \eqref{e:swsys}.}\label{fig:graph}
    \end{figure}
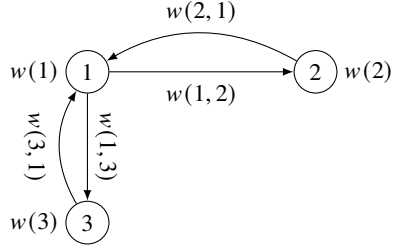
    
    We have that \(1,2,3\) are the vertices of \(\G\) such that \(\C_1=\{W_{1\to 1}^{(1)},W_{1\to 1}^{(2)}\}\neq\emptyset\), \(\C_2 = \{W_{2\to 2}\}\neq\emptyset\), \(\C_3 = \{W_{3\to 3}\}\neq\emptyset\), where \(W_{1\to 1}^{(1)} = 1,2,1\), \(W_{1\to 1}^{(2)} = 1,3,1\), \(W_{2\to 2} = 2,1,2\) and \(W_{3\to 3} = 3,1,3\). We note that \(W_{1\to 1}^{(1)},W_{1\to 1}^{(2)},W_{2\to 2},W_{3\to 3}\) satisfy condition \eqref{e:aux_condn1}. In particular, 
    \begin{itemize}[label=\(\circ\),leftmargin=*]
        \item for \(W_{1\to 1}^{(1)}\): \(w(1)\delta_1+w(2)\Delta_2+w(1,2)+w(2,1)=-9.33<0\),
        \item for \(W_{1\to 1}^{(2)}\): \(w(1)\delta_1+w(3)\Delta_3+w(1,3)+w(3,1)=-8.64<0\),
        \item for \(W_{2\to 2}\): \(w(2)\Delta_2+w(1)\delta_1+w(2,1)+w(1,2)=-9.33<0\),
        \item for \(W_{3\to 3}\): \(w(3)\Delta_3+w(1)\delta_1+w(3,1)+w(1,3)=-8.64<0\).
    \end{itemize} 
    Further, the simple walks \(W_1 = 2,1\) and \(W_2 = 3,1\) are jointly contractive with \(W_{1\to 1}^{(2)}\) and \(W_{1\to 1}^{(1)}\), respectively. Indeed,
    \begin{itemize}[label=\(\circ\),leftmargin=*]
        \item for \(W' = 2,1,3,1\): \(w(2)\Delta_2+w(1)\delta_1+w(3)\Delta_3+w(2,1)+w(1,3)+w(3,1) = -5.677<0\),
        \item for \(W'' = 3,1,2,1\): \(w(3)\Delta_3+w(1)\delta_1+w(2)\Delta_2+w(3,1)+w(1,2)+w(2,1) = -6.41<0\).
    \end{itemize}
    Thus, we have verified conditions \ref{condn1} and \ref{condn2} by using the assertions of Propositions \ref{prop:mainres1} and \ref{prop:mainres2}.
    
    We generate \(10\) different infinite walks on \(\G\) randomly and focus on their corresponding switching signals that obey the given restrictions on dwell times on the subsystems. We pick initial conditions \(x_0\in[-1,+1]\) and exogenous inputs \(v\in[-0.5,+0.5]\) uniformly at random. The state trajectories, \(\biggl(\norm{x(t)}\biggr)_{t\in[0,+\infty[}\), of the switched system \eqref{e:swsys} under these switching signals are illustrated up to time \(t=15\) units of time in Figure \ref{fig:traj}. IOSS of the switched system \eqref{e:swsys} follows.
    \begin{figure}
    \begin{center}
        \includegraphics[scale=0.5]{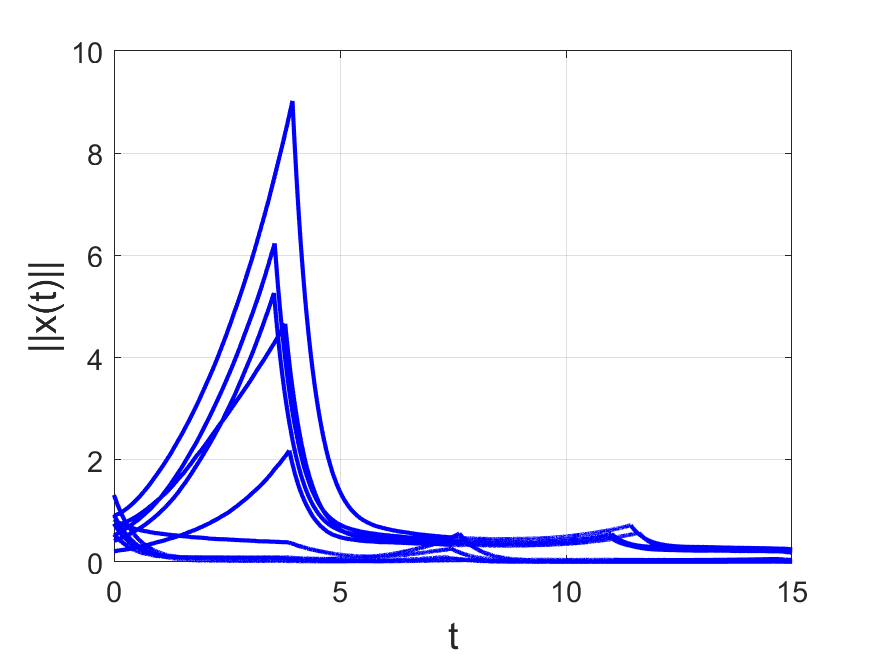}
    \end{center}
    \caption{State trajectories, \(\biggl(\norm{x(t)}\biggr)_{t\in[0,+\infty[}\), of the switched system \eqref{e:swsys}.}\label{fig:traj}
    \end{figure} 
\section{Concluding remarks}
\label{s:concln}
    In this paper we studied IOSS of continuous-time switched nonlinear systems whose switching signals obey pre-specified restrictions on admissible switches between the subsystems and admissible dwell times on the subsystems. We presented sufficient conditions on the subsystems dynamics, the set of admissible switches between the subsystems and the admissible minimum and maximum dwell times on the subsystems under which all the switching signals obeying the given restrictions are stabilizing. Our results do not cater to switching signals that dwell only on unstable subsystems. We identify the problem of IOSS under all switching signals that obey pre-specified restrictions and activate only unstable subsystems as a direction of future work.
\section{Proofs of our results}
\label{s:proofs}
    {\it Proof of Theorem \ref{t:mainres}}:
    Fix \(\sigma\in\SR\) arbitrary. Let \(\overline{W} = \overline{v}_0,\overline{v}_1,\ldots\) be the infinite walk on \(\G\) corresponding to \(\sigma\). The following set of auxiliary results will be useful in the proofs of our results.
    \begin{lemma}
    \label{lem:auxres1}
        Suppose that conditions \ref{condn1} and \ref{condn2} hold. Then \(\overline{W}\) is a concatenation of contractive walks.
    \end{lemma}
    \begin{proof}
        Let \(m\) be the total number of vertices \(v\in\P\) such that \(\C_v\neq\emptyset\). Since \(\G\) has a finite number of vertices, there is at least one vertex \(v\in\P\) that appears in \(\overline{W}\) more than once.
        
        Let \(n_1\in\N\) be the first index satisfying \(\overline{v}_{n_1} = v^{(1)}\) such that \(v^{(1)}\) appears more than once in \(\overline{W}\). Let \(\overline{n}_1\in\N\) be the last index in \(\overline{W}\) such that \(\overline{v}_{\overline{n}_1} = v^{(1)}\). Thus, the subwalk \(\overline{v}_{n_1},\overline{v}_{n_1+1},\ldots,\overline{v}_{\overline{n}_1}\) of \(\overline{W}\) is a concatenation of closed walks \(W_{v^{(1)}\to v^{(1)}}^{(j)}\), \(j\in\{1,2,\ldots,r_{v^{(1)}}\}\), which are all contractive under condition \ref{condn1}. Since \(\G\) has a finite number of vertices, it follows that there is at least one vertex \(v\in\P\), \(v\neq v^{(1)}\), that appears more than once in \(\overline{v}_{\overline{n}_1+1},\overline{v}_{\overline{n}_1+2},\ldots\).\\
        Iterating the above set of arguments, we obtain that \(\overline{W}\) includes concatenation of \(W_{v^{(k)}\to v^{(k)}}^{(j_k)}\), \(j_k\in\{1,2,\ldots,r_{v^{(k)}}\}\), \(k=1,2,\ldots,\overline{r}\), \(\overline{r}\leq m\).
        
        Now, if \(n_1\neq 0\) and \(\overline{v}_0 = u\), then \(W_{u\to v^{(1)}}\) is a subwalk of \(\overline{W}\). In particular, \(W_{u\to v^{(1)}}\) is a simple walk. In view of condition \ref{condn2}, \(W_{u\to v^{(1)}}\) is jointly contractive with \(W_{v^{(1)}\to v^{(1)}}^{(j_1)}\), \(j_1\in\{1,2,\ldots,r_{v^{(1)}}\}\). Consequently, the walk \(\tilde{W}_{u\to v^{(1)}}\) which is a concatenation of \(W_{u\to v^{(1)}}\) and \(W_{v^{(1)}\to v^{(1)}}^{(j_1)}\) is contractive.
        
        Further, \(W_{v^{(1)}\to v^{(2)}},W_{v^{(2)}\to v^{(3)}},\ldots,W_{v^{(\overline{r}-1)}\to v^{(\overline{r})}}\) are possible subwalks of \(\overline{W}\). In view of condition \ref{condn2}, \(W_{v^{(i)}\to v^{(i+1)}}\) is jointly contractive with \(W_{v^{(i+1)}\to v^{(i+1)}}\) leading to contractivity of the walk \(\tilde{W}_{v^{(i)}\to v^{(i+1)}}\) which is a concatenation of \(W_{v^{(i)}\to v^{(i+1)}}\) and \(W_{v^{(i+1)}\to v^{(i+1)}}\), \(i=1,2,\ldots,\overline{r}-1\).
    \end{proof}
    
    Consider an interval \(]s,t]\subseteq[0,+\infty[\) of time. Let \(s=:\overline{\tau}_0<\overline{\tau}_1<\cdots<\overline{\tau}_n\) be the switching instants on \(]s,t]\). Consider the walk \(W = \sigma(\overline{\tau}_0),\sigma(\overline{\tau}_1),\ldots,\sigma(\overline{\tau}_n)\) on \(\G\). Define 
    \begin{align*}
        \overline{\Xi}(W) := \sum_{i=0}^{n-1}w(\sigma(\overline{\tau}_i))\tilde{D}_{\sigma(\overline{\tau}_i)}
        +\sum_{i=0}^{n-1}w(\sigma(\overline{\tau}_i),\sigma(\overline{\tau}_{i+1})),
    \end{align*}
    where \(\tilde{D}_v\in[\delta_{v},\Delta_v]\) is the dwell time of \(\sigma\) on subsystem \(v\).
    
    \begin{lemma}
    \label{lem:auxres3}
         The following is true:
        \begin{align*}
            &-\sum_{p\in\P_S}\abs{\lambda_p}\Tsw_p(s,t)+\sum_{p\in\P_U}\abs{\lambda_p}\Tsw_p(s,t)
            +\sum_{(p,q)\in E(\P)}(\ln\mu_{pq})\Nsw_{pq}(s,t) 
            = \overline{\Xi}(W).
        \end{align*}
    \end{lemma}
    
    \begin{proof}
    We have
    \begin{align*}
        \overline{\Xi}(W) &= \sum_{i=0}^{n-1}w(\sigma(\overline{\tau}_i))\tilde{D}_{\sigma(\overline{\tau}_i)}
        +\sum_{i=0}^{n-1}w(\sigma(\overline{\tau}_i),\sigma(\overline{\tau}_{i+1}))\\
        &=\sum_{p\in\P_S}w(p)\sum_{\substack{i=0\\\sigma(\overline{\tau}_i)=p}}^{n-1}\tilde{D}_p
        +\sum_{p\in\P_U}w(p)\sum_{\substack{i=0\\\sigma(\overline{\tau}_i)=p}}^{n-1}\tilde{D}_p
        +\sum_{(p,q)\in E(\P)}w(p,q)\sum_{\substack{i=0\\\sigma(\overline{\tau}_i)=p\\\sigma(\overline{\tau}_i)=q}}^{n-1}1\\
        &=-\sum_{p\in\P_S}\abs{\lambda_p}\Tsw_p(s,t)+\sum_{p\in\P_U}\abs{\lambda_p}\Tsw_p(s,t)
        +\sum_{(p,q)\in E(\P)}w(p,q)\Nsw_{pq}(s,t).
    \end{align*}
    \end{proof}
    
    \begin{lemma}
    \label{lem:auxres2}
        Let \(W\) be a concatenation of the walks \(W_j = v_0^{(j)},v_1^{(j)},\ldots,v_{n_j}^{(j)}\), \(j=1,2,\ldots,r\), on \(\G\) with \(v_{n_j}^{(j)}=v_0^{(j+1)}\), \(j=1,2,\ldots,r-1\). Consider the segment of \(\sigma\) corresponding to \(W\). We have \(\displaystyle{\overline{\Xi}(W) = \sum_{j=1}^{r}\overline{\Xi}(W_j)}\).
    \end{lemma}
    
    \begin{proof}
     By construction, \(W = v_0^{(1)},v_1^{(1)},\ldots,v_{n_1}^{(1)},v_1^{(2)},v_2^{(2)},\ldots,
    v_{n_2}^{(2)},\ldots,v_{n_{r-1}}^{(r-1)},v_1^{(r)},v_2^{(r)},\ldots,v_{n_r}^{(r)}\). In view of \(v_{n_j}^{(j)} = v_0^{(j+1)}\), \(j=1,2,\ldots,r-1\), we have 
    \begin{align*}
        \overline{\Xi}(W) &= \sum_{i=0}^{n_1-1}w(v_i^{(1)})\tilde{D}_{v_i^{(1)}}+\sum_{i=0}^{n_1-1}w(v_i^{(1)},v_{i+1}^{(1)})+\sum_{i=0}^{n_2-1}w(v_i^{(2)})\tilde{D}_{v_i^{(2)}}+\sum_{i=0}^{n_2-1}w(v_i^{(2)},v_{i+1}^{(2)})+\cdots+\\
        &\quad\quad+\sum_{i=0}^{n_r-1}w(v_i^{(r)})\tilde{D}_{v_i^{(r)}}+\sum_{i=0}^{n_r-1}w(v_i^{(r)},v_{i+1}^{(r)})\\
        &=\overline{\Xi}(W_1)+\overline{\Xi}(W_2)+\cdots+\overline{\Xi}(W_r).
    \end{align*}
    \end{proof}
    
    We are now in a position to present our proof of Theorem \ref{t:mainres}.
    
    {\it Proof of Theorem \ref{t:mainres}}: Fix \(t > 0\). Recall that \(0=:\tau_{0}<\tau_{1}<\cdots<\tau_{\Nsw(0,t)}\) are the switching instants before (and including) \(t\). In view of \eqref{e:key2}, we have that
        \begin{align*}
            V_{\sigma(t)}(x(t))&\leq\exp\bigl(-\lambda_{\sigma(\tau_{\Nsw(0,t)})}(t-\tau_{\mathrm{N}(0,t)})\bigr)V_{\sigma(t)}(x(\tau_{\Nsw(0,t)}))\nn\\
            &\:\:\:\:+\bigl(\gamma_{1}(\norm{v}_{[0,t]})+\gamma_{2}(\norm{y}_{[0,t]})\bigr)\times
            \int_{\tau_{\Nsw(0,t)}}^{t}\exp\bigl(-\lambda_{\sigma(\tau_{\Nsw(0,t)})}(t-s)\bigr)ds.
        \end{align*}
        An iteration of the above with an application of \eqref{e:key3} leads to
        \begin{align}
        \label{e:pf1_step1}
            V_{\sigma(t)}(x(t))&\leq\psi_{1}(t)V_{\sigma(0)}(x_{0})
           +\bigl(\gamma_{1}(\norm{v}_{[0,t]})+\gamma_{2}(\norm{y}_{[0,t]})\bigr)\psi_{2}(t),
        \end{align}
        where
        \begin{align}
        \label{e:psi1_defn}
            \psi_{1}(t) := \exp\biggl(-&\sum_{\substack{i=0\\\tau_{\Nsw(0,t)+1}:=t}}^{\Nsw(0,t)}\lambda_{\sigma(\tau_{i})}(\tau_{i+1}-\tau_{i})+\sum_{i=0}^{\Nsw(0,t)-1}\ln\mu_{\sigma(\tau_{i})\sigma(\tau_{i+1})}\biggr),
        \end{align}
        and
        \begin{align}
        \label{e:psi2_defn}
            \psi_{2}(t) &:= \sum_{\substack{i=0\\\tau_{\Nsw(0,t)+1}:=t}}^{\Nsw(0,t)}\Biggl(\exp\Biggl(-\sum_{\substack{j=i+1\\\tau_{\Nsw(0,t)+1}:=t}}^{\Nsw(0,t)}\lambda_{\sigma(\tau_{j})}(\tau_{j+1}-\tau_{j})
            +\sum_{j=i+1}^{\Nsw(0,t)-1}\ln\mu_{\sigma(\tau_{j})\sigma(\tau_{j+1})}\Biggr)\\
            &\quad\quad\quad\quad\times\frac{1}{\lambda_{\sigma(\tau_{i})}}\bigl(1-\exp\bigl(-\lambda_{\sigma(\tau_{i})}(\tau_{i+1}-\tau_{i})\bigr)\bigr)\Biggr).
        \end{align}
        An application of \eqref{e:key1} to \eqref{e:pf1_step1} leads to
        \begin{align*}
            \underline{\alpha}(\norm{x(t)})&\leq\psi_{1}(t)\overline{\alpha}(\norm{x_{0}})
            +\bigl(\gamma_{1}(\norm{v}_{[0,t]})+\gamma_{2}(\norm{y}_{[0,t]})\bigr)\psi_{2}(t).
        \end{align*}
        In view of Definition \ref{d:ioss}, for IOSS of the switched system \eqref{e:swsys}, we need to show that
        \begin{enumerate}[label=\roman*),leftmargin=*]
            \item\label{verify-1} \(\overline{\alpha}(*)\psi_{1}(\cdot)\) can be bounded above by a class \(\KL\) function, and
            \item\label{verify-2} \(\psi_{2}(\cdot)\) is bounded above by a constant.
        \end{enumerate}
        
        We first verify i). We already have from Assumption \ref{assump:key1} that \(\overline\alpha\in\Kinfty\). Therefore, it remains to show that \(\psi_1(\cdot)\) is bounded above by a class \(\mathcal{L}\) function. We have
        \begin{align}
        \label{e:pf1_step2}
            &-\sum_{\substack{i=0\\\tau_{\Nsw(0,t)+1}:=t}}^{\Nsw(0,t)}\lambda_{\sigma(\tau_i)}(\tau_{i+1}-\tau_{i})
            +\sum_{i=0}^{\Nsw(0,t)-1}\ln\mu_{\sigma(\tau_i)\sigma(\tau_{i+1})}\nonumber\\
            =&-\sum_{p\in\P}\lambda_p\Tsw_p(0,t)+\sum_{(p,q)\in E(\P)}(\ln\mu_{pq})\Nsw_{pq}(0,t)\nonumber\\
            =&-\sum_{p\in\P_S}\abs{\lambda_p}\Tsw_p(0,t)+\sum_{p\in\P_U}\abs{\lambda_p}\Tsw_p(0,t)
            +\sum_{(p,q)\in E(\P)}(\ln\mu_{pq})\Nsw_{pq}(0,t).
        \end{align}
        Recall that \(\sigma\) corresponds to the infinite walk \(\overline{W}\) on \(\G\). In view of Lemma \ref{lem:auxres1}, \(\overline{W}\) is a concatenation of contractive walks on \(\G\). Let the subwalk \(\tilde{W}=\sigma(\tau_0),\sigma(\tau_1),\ldots,\sigma(\tau_{\Nsw(0,t)})\) of \(\overline{W}\) be a concatenation of the contractive walks \(W_1,W_2,\ldots,W_r\) and an initial segment of the contractive walk \(W_{r+1}\). For \(W_j = v_0^{(j)},v_1^{(j)},\ldots,v_{n_j-1}^{(j)},v_{n_j}^{(j)}\), \(j=1,2,\ldots,r+1\), with \(v_{n_j}^{(j)}=v_0^{(j+1)}\), \(j=1,2,\ldots,r\), we define \(\displaystyle{D_{W_j}=\sum_{i=0}^{n_j-1}\tilde{D}_{v_i^{(j)}}}\), where \(\tilde{D}_{v_i}^{(j)}\in[\delta_{v_i^{(j)}},\Delta_{v_i^{(j)}}]\) is the duration of time that \(\sigma\) dwelt on subsystem \(v_{i}^{(j)}\). Define \(\displaystyle{\hat{D}_{\tilde{W}} = \max_{j=1,2,\ldots,r}D_{W_j}}\). Let \(\overline{\Xi}(W_j)=-\varepsilon_j\), \(j=1,2,\ldots,r\). Let \(\overline\ell\) be the maximum length of a walk on \(\G\) that is not contractive. Define \(\displaystyle{\hat{w} = \max_{(u,v)\in E(\P)}\bigl(w(u)D_u+w(u,v)\bigr)}\), where \(D_u=\delta_u\), if \(u\in\P_S\) and \(D_u=\Delta_u\), if \(u\in\P_U\). Then in view of Lemma \ref{lem:auxres2}, we have that
        \begin{align*}
            \overline{\Xi}(W) = \sum_{j=1}^{r}\overline{\Xi}(W_j) + \overline{\Xi}(W_{r-1})\leq-\sum_{j=1}^{r}\varepsilon_j+\overline\ell\hat{w}.
        \end{align*}
        Rewrite \(\overline\ell\hat{w}=\kappa\). It also follows from \eqref{e:psi1_defn}, \eqref{e:pf1_step2} and Lemma \ref{lem:auxres3} that 
        \(\displaystyle{\psi_1(t)\leq\exp\Bigl(-\sum_{j=1}^{r}\varepsilon_j+\kappa\Bigr)}\). Let \(\varphi:[0,t]\to\R\) be a function connecting \((0,\exp(\kappa)+\hat{D}_{\tilde{W}})\), \((D_{W_1},\exp(-\varepsilon_1+\kappa)\),
        \((D_{W_1}+D_{W_2},\exp(-\varepsilon_1-\varepsilon_2+\kappa),\ldots\), \((D_{W_1}+D_{W_2}+\cdots+D_{W_r},\exp(-\varepsilon_1-\varepsilon_2-\cdots-\varepsilon_{r-1}+\kappa))\),
        \((t,\exp(-\varepsilon_1-\varepsilon_2-\cdots-\varepsilon_{r}+\kappa))\) with straight line segments. By construction, \(\varphi\) is an upper envelope of \(T\mapsto\varphi_1(T)\) on \([0,t]\), is continuous, decreasing and tends to \(0\) as \(t\to+\infty\). Consequently, \(\varphi\) is a class \(\mathcal{L}\) function.
        
        We next verify ii). The function \(\psi_2(t)\) is
        \begin{align}
        \label{e:pf1_step3}
            &\sum_{p\in\P_{S}}\frac{1}{\abs{\lambda_{p}}}\sum_{\substack{i=0\\\sigma(\tau_{i})=p\\
            \tau_{\Nsw(0,t)+1}:=t}}^{\Nsw(0,t)}\Biggl(\exp\Biggl(-\sum_{k\in\P_{S}}\abs{\lambda_{k}}\Tsw_{k}(\tau_{i+1},t)
            +\sum_{k\in\P_{U}}\abs{\lambda_{k}}\Tsw_{k}(\tau_{i+1},t)\nn\\
            &\hspace*{1cm}+\sum_{(k,\ell)\in E(\P)}(\ln\mu_{k\ell})\Nsw_{k\ell}(\tau_{i+1},t)\Biggr)
            \biggl(1-\exp(-\abs{\lambda_{p}}(\tau_{i+1}-\tau_{i}))\biggr)\Biggr)\nn\\
            -&\sum_{p\in\P_{U}}\frac{1}{\abs{\lambda_{p}}}\sum_{\substack{i=0\\
            \sigma(\tau_i)=p\\\tau_{\Nsw(0,t)+1}:=t}}^{\Nsw(0,t)}\Biggl(\exp\Biggl(-\sum_{k\in\P_{S}}\abs{\lambda_{k}}\Tsw_{k}(\tau_{i+1},t)+\sum_{k\in\P_{U}}\abs{\lambda_{k}}\Tsw_{k}(\tau_{i+1},t)\nn\\
            &\hspace*{1cm}+\sum_{(k,\ell)\in E(\P)}(\ln\mu_{k\ell})\Nsw_{k\ell}(\tau_{i+1},t)\Biggr)
           \biggl(1-\exp(\abs{\lambda_{p}}(\tau_{i+1}-\tau_{i}))\biggr)\Biggr)\nn\\
            \leq&\sum_{p\in\P_{S}}\frac{1}{\abs{\lambda_{p}}}
            \sum_{\substack{i=0\\\sigma(\tau_i)=p\\\tau_{\Nsw(0,t)+1}:=t}}^{\Nsw(0,t)}\Biggl(\exp\Biggl(-\sum_{k\in\P_{S}}\abs{\lambda_{k}}\Tsw_{k}(\tau_{i+1},t)\nn\\
            &\hspace*{1cm}+\sum_{k\in\P_{U}}\abs{\lambda_{k}}\Tsw_{k}(\tau_{i+1},t)+\sum_{(k,\ell)\in E(\P)}(\ln\mu_{k\ell})\Nsw_{k\ell}(\tau_{i+1},t)\Biggr)\Biggr)\nn\\
            +&\sum_{p\in\P_{U}}\frac{1}{\abs{\lambda_{p}}}\sum_{\substack{i=0\\\sigma(\tau_i)=p\\\tau_{\Nsw(0,t)+1}:=t}}^{\Nsw(0,t)}\Biggl(\exp\Biggl(-\sum_{k\in\P_{S}}\abs{\lambda_{k}}\Tsw_{k}(\tau_{i},t)\nn\\
            &\hspace*{1cm}+\sum_{k\in\P_{U}}\abs{\lambda_{k}}\Tsw_{k}(\tau_{i},t)+\sum_{(k,\ell)\in E(\P)}(\ln\mu_{k\ell})\Nsw_{k\ell}(\tau_{i},t)\Biggr)\Biggr)\nn\\
            \leq&\sum_{p\in\P_{S}}\frac{1}{\abs{\lambda_{p}}}\sum_{i=0}^{\Nsw(0,t)}\Biggl(\exp\Biggl(-\sum_{k\in\P_{S}}
            \abs{\lambda_{k}}\Tsw_{k}(\tau_{i+1},t)\nn\\
            &\hspace*{1cm}+\sum_{k\in\P_{U}}\abs{\lambda_{k}}\Tsw_{k}(\tau_{i+1},t)+\sum_{(k,\ell)\in E(\P)}(\ln\mu_{k\ell})\Nsw_{k\ell}(\tau_{i+1},t)\Biggr)\Biggr)\nn\\
            +&\sum_{p\in\P_{U}}\frac{1}{\abs{\lambda_{p}}}\sum_{i=0}^{\Nsw(0,t)}\Biggl(\exp\Biggl(-\sum_{k\in\P_{S}}
            \abs{\lambda_{k}}\Tsw_{k}(\tau_{i},t)\nn\\
            &\hspace*{1cm}+\sum_{k\in\P_{U}}\abs{\lambda_{k}}\Tsw_{k}(\tau_{i},t)+\sum_{(k,\ell)\in E(\P)}(\ln\mu_{k\ell})\Nsw_{k\ell}(\tau_{i},t)\Biggr)\Biggr).
        \end{align}
        We are interested in boundedness of the sum
        \begin{align*}
            &\sum_{i=0}^{\Nsw(0,t)}\Biggl(\exp\Biggl(-\sum_{k\in\P_{S}}
            \abs{\lambda_{k}}\Tsw_{k}(\tau_{i},t)
            +\sum_{k\in\P_{U}}\abs{\lambda_{k}}\Tsw_{k}(\tau_{i},t)
            +\sum_{(k,\ell)\in E(\P)}(\ln\mu_{k\ell})\Nsw_{k\ell}(\tau_{i},t)\Biggr).
        \end{align*}
        Let 
        \begin{align*}
            \Theta(s,t)&:= -\sum_{k\in\P_{S}}
            \abs{\lambda_{k}}\Tsw_{k}(\tau_{i},t)
            +\sum_{k\in\P_{U}}\abs{\lambda_{k}}\Tsw_{k}(\tau_{i},t)
            +\sum_{(k,\ell)\in E(\P)}(\ln\mu_{k\ell})\Nsw_{k\ell}(\tau_{i},t).
        \end{align*}
        It follows that
        \begin{align*}
            &\sum_{i=0}^{\Nsw(0,t)}\exp\Bigl(\Theta(\tau_i,t)\Bigr) = \exp\Bigl(\Theta(\tau_0,t)\Bigr) 
            +\exp\Bigl(\Theta(\tau_1,t)\Bigr)+\cdots+\exp\Bigl(\Theta(\tau_{n_r},t)\Bigr)\\
            &\hspace*{3cm}+\cdots+
            \exp\Bigl(\Theta(\tau_\Nsw(0,t),t)\Bigr)\\
            =&\exp(-(\varepsilon_1+\varepsilon_2+\cdots+\varepsilon_r))+
            \exp(-(\varepsilon_1+\varepsilon_2+\cdots+\varepsilon_r)+\hat{w})\\
            &\quad+ \exp(-(\varepsilon_1+\varepsilon_2+\cdots+\varepsilon_r)+2\hat{w})+\cdots
            + \exp(-(\varepsilon_1+\varepsilon_2+\cdots+\varepsilon_r)+(n_1-1)\hat{w})\\
            &\quad+\exp(-(\varepsilon_2+\cdots+\varepsilon_r))+
            \exp(-(\varepsilon_2+\cdots+\varepsilon_r)+\hat{w})
            + \exp(-(\varepsilon_2+\cdots+\varepsilon_r)+2\hat{w})+\cdots\\
            &\quad+ \exp(-(\varepsilon_2+\cdots+\varepsilon_r)+(n_1-1)\hat{w})
            +\cdots+\exp(\hat{w})+\exp(2\hat{w})+\cdots+\exp(\overline{\ell}\hat{w})\\
            =&\exp(-(\varepsilon_1+\varepsilon_2+\cdots+\varepsilon_r))
            (1+\exp(\hat{w})+\exp(2\hat{w})+\cdots+\exp((n_1-1)\hat{w}))\\
            &\quad+\exp(-(\varepsilon_2+\cdots+\varepsilon_r))
            (1+\exp(\hat{w})+\exp(2\hat{w})+\cdots+\exp((n_2-1)\hat{w}))\\
            &\quad+(\exp(\hat{w})+\exp(2\hat{w})+\cdots+\exp(\overline{\ell}\hat{w}))\\
            =&\exp(-(\varepsilon_1+\varepsilon_2+\cdots+\varepsilon_r))\frac{\exp(n_1\hat{w})-1}{\exp(\hat{w})-1}
            +\exp(-(\varepsilon_2+\cdots+\varepsilon_r))\frac{\exp(n_2\hat{w})-1}{\exp(\hat{w})-1}\\
            &\hspace*{2cm}+\cdots+\frac{\exp((\overline{\ell}+1)\hat{w})-1}{\exp(\hat{w})-1}-1.
        \end{align*}
        Similarly, boundedness of the sum 
        \begin{align*}
            &\sum_{i=0}^{\Nsw(0,t)}\Biggl(\exp\Biggl(-\sum_{k\in\P_{S}}
            \abs{\lambda_{k}}\Tsw_{k}(\tau_{i+1},t)
            +\sum_{k\in\P_{U}}\abs{\lambda_{k}}\Tsw_{k}(\tau_{i+1},t)          
            +\sum_{(k,\ell)\in E(\P)}(\ln\mu_{k\ell})\Nsw_{k\ell}(\tau_{i+1},t)\Biggr)
        \end{align*}
        can be established. Since the sets \(\P_S\) and \(\P_U\) are finite, ii) follows.
        
        This completes our proof of Theorem \ref{t:mainres}.\hspace*{2.4cm}\(\square\)
        
    {\it Proof of Proposition \ref{prop:mainres1}}:
    We have
    \begin{align*}
        \Xi(W) &:= \sum_{i=0}^{n-1}w(v_i)D_{v_i} + \sum_{i=0}^{n-1}w(v_i,v_{i+1})\\
        &= \sum_{\substack{i=0\\v_i\in\P_S}}^{n-1}w(v_i)D_{v_i} + \sum_{\substack{i=0\\v_i\in\P_U}}^{n-1}w(v_i)D_{v_i} + \sum_{i=0}^{n-1}w(v_i,v_{i+1})\\
        &\leq -\sum_{\substack{i=0\\v_i\in\P_S}}^{n-1}\abs{\lambda_{v_i}}\delta_{v_i}+
        \sum_{\substack{i=0\\v_i\in\P_U}}^{n-1}\abs{\lambda_{v_i}}\Delta_{v_i}+
        \sum_{i=0}^{n-1}w(v_i,v_{i+1})\\
        &= \sum_{\substack{i=0\\v_i\in\P_S}}^{n-1}w(v_i)\delta_{v_i} + \sum_{\substack{i=0\\v_i\in\P_U}}^{n-1}w(v_i)\Delta_{v_i} + \sum_{i=0}^{n-1}w(v_i,v_{i+1}).
    \end{align*}
    In view of \eqref{e:aux_condn1}, the right-hand side of the above expression is strictly less than \(0\). It follows that \(W\) is contractive.\hspace*{1.5cm}\( \square\)
    
    {\it Proof of Proposition \ref{prop:mainres2}}: Consider a closed walk \(W = v_0,v_1,\ldots,v_{n-1},v_n\) on \(\G\) with \(v_0=v_n=v\). 
    
    If \(W\) is a cycle, then it is contractive by hypothesis of Proposition \ref{prop:mainres2}. We will show that if \(W\) is not a cycle, then it can be written as a concatenation of cycles. By definition of a closed walk that is not a cycle, there is at least one vertex \(\overline{v}\in\P\) such that \(\overline{v}\) appears more than once in \(W\).
    
    We decompose \(W\) into the following subwalks: \(\tilde{W}_{v_0\to v_{n_1}},\tilde{W}_{v_{n_1}\to v_{\overline{n}_1}},\tilde{W}_{v_{\overline{n}_1}\to v_{{n}_2}},\tilde{W}_{v_{n_2}\to v_{\overline{n}_2}},\ldots\),\\\(\tilde{W}_{v_{n_r}\to v_{\overline{n}_r}},\tilde{W}_{v_{\overline{n}_r}\to v_{0}}\), where \(n_1\) satisfying \(1\leq n_1\leq n-1\) is the first index where a vertex \(\overline{v}^{(1)}\in\P\) that is repeated in \(W\) appears first in \(W\), 
    \(\overline{n}_1\) satisfying \(n_1<\overline{n}_1<n-1\) is the last index where \(\overline{v}^{(1)}\) appears in \(W\),
    \(n_2\) satisfying \(\overline{n}_1<n_2< n-1\) is the first index where a vertex \(\overline{v}^{(2)}\in\P\) that is repeated in the subwalk \(\tilde{W}_{v_{\overline{n}_1+1}\to v_{n-1}}\) appears first in  \(\tilde{W}_{v_{\overline{n}_1+1}\to v_{n-1}}\), \(\overline{n}_2\) satisfying \(n_2<\overline{n}_2< n-1\) is the last index where \(\overline{v}^{(2)}\) appears in \(\tilde{W}_{v_{\overline{n}_1+1}\to v_{n-1}}\), \ldots, 
    \(n_r\) satisfying \(\overline{n}_{r-1}<n_r< n-1\) is the first index where a vertex \(\overline{v}^{(r)}\in\P\) that is repeated in the subwalk \(\tilde{W}_{v_{\overline{n}_{r-1}+1}\to v_{n-1}}\) appears first in  \(\tilde{W}_{v_{\overline{n}_{r-1}+1}\to v_{n-1}}\), \(\overline{n}_r\) satisfying \(n_r<\overline{n}_r\leq n-1\) is the last index where \(\overline{v}^{(r)}\) appears in \(W\). Clearly, \(\overline{\overline{W}} = v_0,v_1,\ldots,v_{n_1},v_{\overline{n}_1+1},v_{\overline{n}_1+2},\ldots,v_{n_2},
    v_{\overline{n}_2+1},v_{\overline{n}_2+2},\ldots,v_{n_r}\),\\\(
    v_{\overline{n}_r+1},v_{\overline{n}_r+2},\ldots,v_n\) is a cycle. Now, if all the closed walks \(\tilde{W}_{v_{n_j}\to v_{\overline{n}_j}}\), \(j=1,2,\ldots,r\) are cycles, then they are contractive by hypothesis.
    
    Consider a closed walk \(\tilde{W}_{v_{n_j}\to v_{\overline{n}_j}}\), \(j\in\{1,2,\ldots,r\}\) that is not a cycle. Let \(n'\) satisfying \(n_j<n'<\overline{n}_j\) be the last index where a vertex \(\overline{v}\in\P\) that appears more than once in \(W\), appears with exactly one appearance left in \(\tilde{W}_{v_n'\to v_{\overline{n}_j}}\). Let \(n''\) satisfying \(n'<n''<\overline{n}_j\) be the index where \(\tilde{v}\) appears last. Then \(W_{v_{n'}\to v_{n''}}\) is a cycle. Indeed, otherwise, the definition of \(n'\) is violated. Also, \(W' = v_{n_j},v_{n_j+1},\ldots,v_{n'},v_{n''+1},\ldots,v_{\overline{n}_j}\) is a closed walk. Iterating the above, we arrive at a cycle \(W''\). Indeed, it follows from the fact that \(\tilde{W}_{v_{n_j}\to v_{\overline{n}_j}}\) is a closed walk. Consequently, \(\tilde{W}_{v_{n_j}\to v_{\overline{n}_j}}\) is a concatenation of cycles, which by hypothesis of Proposition \ref{prop:mainres2} are contractive. In view of Lemma \ref{lem:auxres1} it follows that \(\tilde{W}_{v_{n_j}\to v_{\overline{n}_j}}\) is a contractive closed walk.
    
    In view of contractivity of \(\overline{\overline{W}}\) and the closed walks \(\tilde{W}_{v_{n_j}\to v_{\overline{n}_j}}\), \(j=1,2,\ldots,r\), we have that the closed walk \(W\) is contractive.


\end{document}